\documentclass[11pt]{article}

\usepackage[margin=1in]{geometry}
\usepackage{amsmath, amssymb, amsthm, mathtools}

\usepackage{aliascnt}

\newtheorem{theorem}{Theorem}[section]

\newaliascnt{lemma}{theorem}
\newtheorem{lemma}[lemma]{Lemma}
\aliascntresetthe{lemma}

\newtheorem{proposition}[theorem]{Proposition}

\newtheorem{claim}[theorem]{Claim}

\newtheorem{definition}[theorem]{Definition}

\newtheorem*{theorem*}{Theorem}

\newtheorem*{corollary*}{Corollary}
\newtheorem*{conjecture*}{Conjecture}
\newtheorem*{lemma*}{Lemma}
\newtheorem*{thm*}{Theorem}
\newtheorem*{prop*}{Proposition}
\newtheorem*{obs*}{Observation}
\newtheorem*{definition*}{Definition}

\newtheorem*{remark*}{Remark}
\newtheorem*{rec*}{Recommendation}

\usepackage{enumitem}
\usepackage{graphicx}
\usepackage{fancybox}
\usepackage{comment}
\usepackage{xcolor}
\usepackage{nameref}
\usepackage{bbm}

\definecolor{ForestGreen}{rgb}{0.1333,0.5451,0.1333}
\definecolor{DarkRed}{rgb}{0.8,0,0}
\definecolor{Red}{rgb}{1,0,0}
\usepackage[linktocpage=true,
pagebackref=true,colorlinks,
linkcolor=ForestGreen,citecolor=ForestGreen,
bookmarks,bookmarksopen,bookmarksnumbered]
{hyperref}

\usepackage{amsmath}
\usepackage{amssymb}
\usepackage{amsthm}
\usepackage{bm}

\usepackage{algorithm}
\usepackage{algpseudocode}

\makeatletter
\renewcommand{\theHALG@line}{\thealgorithm.\arabic{ALG@line}}
\makeatother

\usepackage{subfig}

\usepackage{thm-restate}

\usepackage{float}
\usepackage{cleveref}

\renewcommand{\deg}{\textnormal{deg}}

\newcommand\R{\mathbb{R}}

\newcommand{\E}[1]{\mathop{{}\mathbb{E}}\left[#1\right]}

\renewcommand{\epsilon}{\ensuremath\varepsilon}

\renewcommand{\phi}{\ensuremath{\varphi}}

\providecommand{\Prb}{\mathbb{P}}

\title{VAC: A \underline{V}olume-sampling-based Elimination Rule for \underline{A}pproximate \underline{C}holesky Factorization}
\author{Yves Baumann, Rasmus Kyng\thanks{The research leading to these results has received funding from the starting grant “A New Paradigm for Flow and Cut Algorithms” (no. TMSGI2 218022) and grant no. 200021 204787 of the Swiss National Science Foundation.} , Gernot Zöcklein\footnotemark[1]}
\date{\today}

\begin{document}
\maketitle

\begin{abstract}
We propose Volume Appproximate Cholesky (VAC), an alternative sampling rule for practical approximate Cholesky algorithms. Our rule samples a uniformly random spanning tree of the arising product clique to reduce the fill-in generated at each step. Sampling a random spanning tree preserves the edgewise marginals of the provably correct scheme of Kyng-Sachdeva (FOCS`16), while ensuring connectivity in the spirit of the practical rule proposed in Gao-Kyng-Spielman (SISC`26). Our sampling method is simple, provably linear time and also admits a $O(\log n)$ depth parallel implementation.
\end{abstract}

\section{Preliminaries}

The main object of study will be \emph{product cliques}. These are the family of graphs that arise as a fill-in when implementing a cholesky factorization of the Laplacian $L$ of a graph. A line of work has shown both theoretical guarantees as well as practical performance for \emph{Approximate Cholesky} algorithms, where the dense fill-in cliques are replaced by sparse sub samples \cite{KS16, GKS23}. In this note, we show how to sample a uniformly random spanning tree of a product clique.

To start, let $V=\{1,\ldots,n\}$ with $n\ge 2$, and let
$w:V\to\R_{>0}$ be a positive vertex-weight function.  Write
\[
    W:=\sum_{x\in V}w(x).
\]

\begin{definition}[Product clique]
The \emph{product clique} $K_w$ is the complete graph on $V$ with edge
conductances
\[
    c_{uv}:=w(u)w(v),
    \qquad u\ne v.
\]
\end{definition}

For an edge $e=(u,v)$, let
\[
    b_e:=e_u-e_v.
\]
If $L$ is the Laplacian of $K_w$, the leverage score of edge $e = (u, v)$ is defined as
\[
    \tau_{uv}:=c_{uv} \cdot b_e^{\top}L^+b_e.
\]
Note that we can write the Laplacian as
\[
    L=W\operatorname{diag}(w)-ww^{\mathsf T}.
\]

We will need the standard Matrix Determinant lemma. See for example \cite{horn} for a proof.
\begin{lemma}[Matrix Determinant Lemma]
    Let $A \in \R^{n \times n}$ be an invertible matrix, and let $u, v \in \R^n$. Then 
    $$
    \det(A + uv^{\top}) = (1 + u^{\top} A^{-1} v) \det(A).
    $$
\end{lemma}

\section{The weighted spanning-tree distribution}

For a spanning tree $T$ of a graph with edge conductance $c \in \R^E$, write
\[
    c(T):=\prod_{e\in T}c_e.
\]
Let $\mathcal{T}$ be a random spanning tree. The weighted uniform spanning-tree distribution is
\[
    \Prb[\mathcal T=T]
    =\frac{c(T)}{\sum_{T'}c(T')}.
\]
In the product clique,
\[
    c(T)
    =\prod_{\{u,v\}\in T}w(u)w(v)
    =\prod_{x\in V}w(x)^{\deg_T(x)}.
\]

\begin{lemma}\label{lem:st_partition}
The weighted spanning-tree partition function of $K_w$ is
\[
    Z_w:=\sum_T c(T)
    =W^{n-2}\prod_{x\in V}w(x).
\]
Consequently, every fixed spanning tree $T$ satisfies
\[
    \Prb[\mathcal T=T]
    =\frac{1}{W^{n-2}}
      \prod_{x\in V}w(x)^{\deg_T(x)-1}.
\]
\end{lemma}

\begin{proof}
Fix a root $r\in V$.  By the matrix-tree theorem (see, e.g., section $13$ in \cite{SAGT}), $Z_w$ is the determinant of
the principal minor obtained from $L$ by deleting row and column $r$.
Writing $D_{-r}:=\operatorname{diag}(w(x):x\ne r)$ and
$w_{-r}:=(w(x))_{x\ne r}$, the matrix determinant lemma gives
\begin{align*}
    Z_w
    &=\det\bigl(WD_{-r}-w_{-r}w_{-r}^{\mathsf T}\bigr)\\
    &= \det\bigl(WD_{-r} I - WD_{-r} (WD_{-r})^{-1}w_{-r}w_{-r}^{\mathsf T}\bigr)\\
    &= \det(WD_{-r})
      \det \left(I-(WD_{-r})^{-1}w_{-r}w_{-r}^{\mathsf T}\right) \\
    &=\det(WD_{-r})
      \det\left(I-w_{-r}^{\mathsf T}(WD_{-r})^{-1}w_{-r}\right)\\
    &=W^{n-1}\prod_{x\ne r}w(x)
      \left(1-\frac{W-w(r)}{W}\right)\\
    &=W^{n-2}\prod_{x\in V}w(x).
\end{align*}
Dividing $c(T)=\prod_x w(x)^{\deg_T(x)}$ by $Z_w$ gives the claimed law.
\end{proof}

\paragraph{Leverage scores from spanning-tree marginals}

We start with the standard identity
\[
    \Prb[e\in\mathcal T]=\tau_e
\]
in order to relate weighted spanning-tree marginals to leverage scores. See for example \cite{SAGT} for a proof of this identity.

\begin{proposition}
For every distinct $u,v\in V$,
\[
    \tau_{uv}=\frac{w(u)+w(v)}{W}.
\]
\end{proposition}

\begin{proof}
Fix $e=\{u,v\}$.  Contracting $e$ merges $u$ and $v$ into a vertex of weight
$w(u)+w(v)$.  Indeed, for every $x\notin\{u,v\}$, the total conductance from
the contracted vertex to $x$ is
\[
    w(u)w(x)+w(v)w(x)
    =(w(u)+w(v))w(x).
\]
Thus the contracted graph is again a product clique, and we can use Lemma \ref{lem:st_partition} to derive the total weight of spanning trees containing $e$ as
\[
    c_{uv}\,
    W^{n-3}(w(u)+w(v))
    \prod_{x\ne u,v}w(x).
\]
Dividing by the full partition function
\[
    Z_w=W^{n-2}w(u)w(v)\prod_{x\ne u,v}w(x)
\]
gives
\[
    \Prb[e\in\mathcal T]
    =\frac{w(u)+w(v)}{W}.
\]
The spanning-tree marginal identity now yields the formula for $\tau_{uv}$.
\end{proof}

\section{Sampling a Uniform Spanning Tree}
We will use Pr\"ufer codes as a sampling device. See \cite{gross2018graph} for details on the relationship between Prüfer codes and labeled spanning trees of $V$. Every labeled tree on
$V$ corresponds bijectively to a Pr\"ufer sequence
\[
    P=(P_1,\ldots,P_{n-2})\in V^{n-2},
\]
and if $T$ is the tree corresponding to sequence $P$, then
\[
    \deg_T(x)-1
    =|\{t\in\{1,2,\ldots,n-2\}:P_t=x\}|.
\]
Therefore, if the entries of $P$ are sampled independently according to
\[
    \Prb[P_t=x]=\frac{w(x)}{W},
\]
then
\[
    \Prb[P=P(T)]
    =\frac{1}{W^{n-2}}
      \prod_{x\in V}w(x)^{\deg_T(x)-1},
\]
which is exactly the spanning-tree law derived above. This observation leads to the following simple algorithm for sampling a uniform spanning tree from product cliques. Below, $\textsc{DecodePrüfer}$ is any algorithm that decodes a Prüfer sequence into the tree it encodes.

\begin{algorithm}[H]
\caption{Weighted uniform spanning tree of a product clique}
\label{alg:uniform-spanning-tree}
\begin{algorithmic}[1]
\Require $V=\{1,\ldots,n\}$ and positive weights $w(v)$
\State $W\gets\sum_{v\in V}w(v)$
\For{$t=1,\ldots,n-2$}
    \State Independently sample $P_t\in V$ with probability $w(P_t)/W$
\EndFor
\State $T\gets\Call{DecodePrüfer}{P_1,\ldots,P_{n-2}}$
\State For every $e = (u, v) \in T$, assign it conductance $W \frac{w(u) w(v)}{w(u) + w(v)}$
\State \Return $T$
\end{algorithmic}
\end{algorithm}

\begin{lemma}\label{lem:product-clique-sampling}
For the product clique $K_w$, the re-weighted uniform spanning tree $T$ of \Cref{alg:uniform-spanning-tree} satisfies that
$$
\E{L_T} = L.
$$
Moreover, it can be sampled sequentially with $O(n)$ work.  Alternatively, it can be sampled with
$O(n\log n)$ work and $O(\log n)$ parallel depth.
\end{lemma}
\begin{proof}
The expectation statement is true because $\Pr[e \in T] = \tau_e$. Also, clearly the Pr\"ufer sequence can be sampled with the stated bounds. Standard
decoding then takes linear work; for parallel decoding, one may use the algorithm of
\cite{CAMINITI200797}.
\end{proof}

\section{Use in Approximate Cholesky Algorithms}
Sampling a weighted clique typically comes up in approximate Cholesky factorization. Consider a Laplacian $S$ of a graph $H$ and a vertex $v$ of $H$. The standard Cholesky factorization of $S$ is given in \cref{alg:gaussian-elimination-cholesky}

\begin{algorithm}[H]
\caption{Cholesky Factorization}
\label{alg:gaussian-elimination-cholesky}
\begin{algorithmic}[1]
\Require Graph Laplacian $L$ from graph $G=(V,E,w)$, $|V|=n$.
\Ensure Lower-triangular matrix $\mathcal{L}$ such that
        $\mathcal{L}\mathcal{L}^{\top}=L$
\State $S_0 \gets L$
\For{$i=1,\ldots,n-1$}
    \State $\displaystyle
        \ell_i
        \gets
        \frac{1}{\sqrt{S_{i-1}(i,i)}}\,
        S_{i-1}(:,i)
    $
    \State $\displaystyle
        S_i
        \gets
        S_{i-1}-\frac{1}{S_{i-1}(i,i)}S_{i-1}(:,i)S_{i-1}(:,i)^{\top}
    $
\EndFor
\State $\ell_n \gets \mathbf{0}_{n\times 1}$
\State $\displaystyle
    \mathcal{L}
    \gets
    \begin{bmatrix}
        \ell_1 & \cdots & \ell_n
    \end{bmatrix}
$
\State \Return $\mathcal{L}$
\end{algorithmic}
\end{algorithm}

For ease of notation, we define $\operatorname{STAR}(v,S)$ to be the Laplacian of the subgraph of $H$ consisting of edges incident on $v$. We define
\[
\operatorname{CLIQUE}(v,S) = \operatorname{STAR}(v,S) - \frac{1}{S(v,v)}\,S(:,v)S(:,v)^{\top}.
\]

For example, suppose
\[
L=
\begin{pmatrix}
W  & -a^{\top} \\
-a & \operatorname{diag}(a) + L_{-1}
\end{pmatrix}.
\]
Then
\[
\operatorname{STAR}(1,L) =
\begin{pmatrix}
W  & -a^{\top} \\
-a & \operatorname{diag}(a)
\end{pmatrix}
\qquad\text{and}\qquad
\operatorname{CLIQUE}(1,L)
=
\begin{pmatrix}
0 & 0 \\
0 & \operatorname{diag}(a) - \dfrac{1}{W}aa^{\top}
\end{pmatrix}.
\]

Now, $\operatorname{CLIQUE}(1,L)$, and specifically $aa^\top$ can be dense. This is the reason why Gaussian Elimination can be slow. At each step, one vertex is eliminated, and the clique is added to the remaining subgraph, which can cost up to $O(\operatorname{deg}(v)^2)$ time. However, the term $aa^\top$ has exactly the \emph{product cliques} structure described in the earlier sections. Previous works have studied how one can sparsify this clique in theory \cite{KS16} as well as in practice \cite{GKS23} to find good factorizations of the Laplacian $S$. We propose \cref{lem:product-clique-sampling} as an alternative clique sampling rule, that can be used to sample the resulting cliques in approximate Cholesky. The result is \cref{alg:approximate-cholesky}
\begin{algorithm}[H]
\caption{Volume Approximate Cholesky(s)}
\label{alg:approximate-cholesky}
\begin{algorithmic}[1]
\Require Graph Laplacian $L$ from graph $G=(V,E,w)$, $|V|=n$.
\Ensure Lower-triangular matrix $\mathcal{L}$.
\State $S_0 \gets L$
\State Generate a random permutation $\pi$ of $[n]$
\For{$i=1,\ldots,n-1$}
    \State $\displaystyle
        \ell_i
        \gets
        \frac{1}
        {\sqrt{S_{i-1}\bigl(\pi(i),\pi(i)\bigr)}}
        S_{i-1}\bigl(:,\pi(i)\bigr)
    $
    \State $\displaystyle
        \begin{aligned}
        S_i \gets{}&
        S_{i-1}
        - \operatorname{STAR}\bigl(\pi(i),S_{i-1}\bigr)+
        \Call{WeightedUSTSampling}
        {\ell_i,\,s}
        \end{aligned}
    $
    \Statex \hspace{\algorithmicindent}
        \Comment{See \cref{alg:uniform-spanning-tree}}
\EndFor
\State $\ell_n \gets \mathbf{0}_{n\times 1}$
\State $\displaystyle
    \mathcal{L}
    \gets
    \begin{bmatrix}
        \ell_1 & \cdots & \ell_n
    \end{bmatrix}
$
\State \Return $\mathcal{L}$ and $\pi$
\end{algorithmic}
\end{algorithm}

From \cref{lem:product-clique-sampling}, we know that \cref{alg:uniform-spanning-tree} samples an unbiased estimator of the product clique. Which leads to a factorization that is in expectation the original Laplacian, as was described in the following claim by \cite{GKS23}:
\begin{claim}[{\cite[Claim~3.1]{GKS23}}]\label{clm:approximate-cholesky-unbiased}
When \textsc{ApproximateCholesky} is instantiated with an unbiased clique-sampling routine \textsc{CliqueSample} whose output is positive semidefinite, its output
\[
    \mathcal{L} = \textsc{ApproximateCholesky}(L)
\]
satisfies
\[
    \mathbb{E}\!\left[\mathcal{L}\mathcal{L}^{\top}\right] = L.
\]
\end{claim}

\section{Acknowledgements} 
We thank Anup Rao and David Durfee for helpful discussion and for telling us about a different algorithm for sampling a UST from a product clique. Their algorithm was presented to run in time $O(n^2)$, but can be easily implemented to run in time $O(n \log n)$. 

ChatGPT $5.5$ was used when developing some of the proofs in this article. Codex was used to help with assembling the latex file.

\section{An Independent Work}
A similar elimination rule to VAC, where the dense fill-in clique produced by gaussian elimination is replaced by uniform spanning tree(s), was independently discovered by Chatanya, Musco and Gionis \cite{chaitanya2026cast}.

\bibliographystyle{plain}
\bibliography{refs}

\end{document}